\documentclass[11pt]{article}

\usepackage[normalem]{ulem} 

\usepackage{fullpage}
\usepackage{amssymb}
\usepackage{amsmath}
\usepackage{amsthm}
\usepackage{tikz}
\usepackage{authblk}
\usepackage{enumitem}
\usepackage{algorithm}
\usepackage{blkarray}
\usepackage{thm-restate}
\usetikzlibrary{calc, patterns}

\usepackage[hidelinks, final=true,colorlinks = false,allcolors = {blue}]{hyperref}
\usepackage[capitalize]{cleveref}

\counterwithin{figure}{section}
\counterwithin{table}{section}

\newtheorem{theorem}{Theorem}[section]
\newtheorem{definition}[theorem]{Definition}
\newtheorem{lemma}[theorem]{Lemma}

\newtheorem{proposition}[theorem]{Proposition}
\newtheorem{corollary}[theorem]{Corollary}

\crefname{table}{Table}{Tables}
\crefname{figure}{Figure}{Figures}
\crefname{algorithm}{Algorithm}{Algorithms}
\crefname{theorem}{Theorem}{Theorems}
\crefname{lemma}{Lemma}{Lemmas}
\crefname{definition}{Definition}{Definitions}
\crefname{remark}{Remark}{Remarks}
\crefname{proposition}{Proposition}{Propositions}

\newcommand{\A}{\ensuremath{\mathcal{A}}}

\newcommand{\E}{\ensuremath{\mathbb{E}}}
\newcommand{\N}{\ensuremath{\mathbb{N}}}
\renewcommand{\P}{\ensuremath{\mathbb{P}}}
\newcommand{\R}{\ensuremath{\mathbb{R}}}

\newcommand{\norm}[1]{\ensuremath{\left\|#1\right\|}}

\DeclareMathOperator{\Var}{Var}
\DeclareMathOperator{\Vol}{Vol}
\DeclareMathOperator{\poly}{poly}
\DeclareMathOperator{\conv}{conv}

\DeclareMathOperator{\diam}{diam}

\newcommand{\eps}{\varepsilon}

\title{How to compute the volume in low dimension?}
\author[1,2]{Arjan Cornelissen}
\author[1]{Simon Apers}
\author[3]{Sander Gribling}
\affil[1]{Universit\'e Paris Cit\'e, CNRS, IRIF, Paris, France}
\affil[2]{Simons Institute, UC Berkeley, California, USA}
\affil[3]{Tilburg University, Tilburg, the Netherlands}

\begin{document}
    \maketitle

    \begin{abstract}
        Estimating the volume of a convex body is a canonical problem in theoretical computer science. Its study has led to major advances in randomized algorithms, Markov chain theory, and computational geometry. In particular, determining the query complexity of volume estimation to a membership oracle has been a longstanding open question. Most of the previous work focuses on the high-dimensional limit. In this work, we tightly characterize the deterministic, randomized and quantum query complexity of this problem in the high-precision limit, i.e., when the dimension is constant.
    \end{abstract}

    \section{Introduction}

    The volume estimation problem is a widely-studied problem, that lies on the interface between convex geometry and computational complexity. The typical setting is the following, see for instance~\cite{simonovits2003compute}. We consider a convex body $K \subseteq \R^d$ that satisfies $B_d \subseteq K \subseteq RB_d$, where $B_d = \{x \in \R^d : \norm{x} \leq 1\}$ denotes the unit ball in $d$ dimensions, and $R B_d$ the radius-$R$ ball in $d$ dimensions. Denoting the volume of $K$ by $\Vol(K)$, the aim of volume estimation is to output $\widetilde{V} \geq 0$ such that $|\widetilde{V} - \Vol(K)| < \varepsilon\Vol(K)$, for a given relative precision $\varepsilon > 0$.

    In this work, we are interested in the \emph{query complexity} of volume estimation. Specifically, we will assume access to $K$ through a \textit{membership oracle} which, after querying a point $x \in \R^d$, returns whether $x \in K$. The key question is then how many membership queries (as a function of $d$, $R$ and $\varepsilon$) are required to solve the volume estimation problem. We analyze this question in three \textit{computational models}, namely the deterministic, randomized and quantum models, each of which is more powerful than the last.

    Previous work mostly studied query complexity in the high-dimensional setting, which tracks the dependency of the query complexity on the dimension $d$. The title of this work is a reference to the early survey by Simonovits called ``How to compute the volume in high dimension?''~\cite{simonovits2003compute}, and we refer the interested reader to this survey for a more in-depth discussion of this setting.

    In this work, however, we focus on the low-dimensional limit, or high-precision limit, i.e., the limit where the dimension $d$ is fixed and the precision $\varepsilon$ goes to zero.
    While we are not aware of any literature targeting this regime in the membership oracle setting, there are numerous algorithms in different access models targeting this regime.
    We summarize the most relevant of these below, and refer the interested reader to the excellent survey by Brunel~\cite{brunel2018methods}.

    \paragraph{Low-dimensional computational geometry.}
    Before turning to volume estimation, we discuss a related problem that has been well studied in the low-dimensional limit. In the \textit{convex set estimation} problem, one tries to output a convex set $\widetilde{K} \subseteq \R^d$ that estimates the original convex set $K$.

    One way to measure the error in this problem is through the \textit{relative Nikodym metric}, defined as the volume of the symmetric difference $\widetilde{K} \Delta K$ divided by the volume of $K$. An interesting sequence of results in this setting starts with a paper by Sch\"utt~\cite{schutt1994random}. They considered a bounded convex set $K \subseteq RB_d \subseteq \R^d$ with $R \in O(1)$, and showed that the convex hull of $n$ points sampled uniformly at random from $K$, approximates it up to relative Nikodym distance $O(n^{-2/(d+1)})$. Later, Brunel proved a matching lower bound~\cite{brunel2016adaptive}, implying that in order to achieve relative error~$\varepsilon$, one requires $\Theta(\varepsilon^{-(d+1)/2})$ samples from the uniform distribution over~$K$ to solve the convex set estimation problem. Surprisingly, Baldin and Rei\ss~\cite{baldin2016unbiased} showed that the volume estimation problem is significantly easier in this setting, i.e., they show that $\Theta(\varepsilon^{-2(d+1)/(d+3)})$ uniform samples from $K$ are necessary and sufficient to estimate its volume up to relative error $\varepsilon$.

    Another line of research considers a qualitatively stronger error metric for the convex set estimation problem. To that end, suppose that a convex body $\widetilde{K} \subseteq \R^d$ is an approximation of a convex body $K \subseteq \R^d$. For any $\varepsilon \in (0,1)$, Agarwal, Har-Peled and Varadarajan~\cite{agarwal2004approximating} define $\widetilde{K}$ to be an $\varepsilon$-kernel of $K$ if $\widetilde{K} \subseteq K$, and
    \[\forall u \in S_{d-1}, \qquad \max_{x \in \widetilde{K}} u^Tx - \min_{x \in \widetilde{K}} u^Tx \geq (1-\varepsilon)\left[\max_{x \in K} u^Tx - \min_{x \in K} u^Tx\right],\]
    where $S_{d-1} = \partial B_d$ is the unit sphere in $d$ dimensions. If $\widetilde{K}$ is an $\varepsilon$-kernel of $K$, then the Hausdorff distance between $\widetilde{K}$ and $K$, i.e., the maximum distance between a point in $K$ from $\widetilde{K}$ and vice versa, is at most $\varepsilon\diam(K)$ (see \cref{lem:eps-kernel-hausdorff}). Being an $\varepsilon$-kernel is invariant under taking arbitrary linear transformations, just like the Nikodym metric, and any $\varepsilon$-kernel of $K$ approximates it up to relative Nikodym distance $O(\varepsilon)$ too (see \cref{prop:eps-kernel-volumetric-approximation}), yet the converse does not hold.

    The starting point towards algorithmically constructing such $\varepsilon$-kernels is an early paper by Dudley~\cite{dudley1974metric,har2019proof}, which proves that every convex body $K \subseteq B_d$ has an approximate polytope $\widetilde{K} \subseteq K$ with Hausdorff distance at most $\varepsilon$ and only $O(\varepsilon^{-(d-1)/2})$ faces. Independently, Brohnsteyn and Ivanov~\cite{bronshteyn1975approximation} proved that a similar approximate polytope exists with only $O(\varepsilon^{-(d-1)/2})$ vertices. Subsequently, in the more restricted setting where $K$ is well-rounded, i.e., $B_d \subseteq K \subseteq RB_d$, Agarwal, Har-Peled and Varadarajan~\cite{agarwal2004approximating} used Brohnsteyn and Ivanov's construction to produce an $\varepsilon$-kernel with $O((R/\varepsilon)^{(d-1)/2})$ vertices, and their algorithm was subsequently improved independently by \cite{chan2006faster} and \cite{yu2008practical}. The referenced resources only explicitly analyze their algorithms' performance in terms of the runtime in the case where $K$ is the convex hull of $n$ points, with \cite{yu2008practical} ultimately achieving a runtime of $\widetilde{O}(n + \varepsilon^{-1/2})$ for $d = 2$ and $\widetilde{O}(n + \varepsilon^{d-2})$ for $d \geq 3$. In this work, we build on these algorithmic ideas and port them to the membership oracle setting.

    \paragraph{Our results.}

    In this paper, we tightly characterize the membership oracle query complexity for the convex set estimation and volume estimation problems, in the low-dimensional limit, and in the deterministic, randomized and quantum models. The resulting complexities are displayed in \cref{tbl:results}.

    \begin{table}[!ht]
        \centering
        \begin{tabular}{r|cc|c}
            Problem & \multicolumn{2}{c|}{Convex set estimation} & Volume estimation \\
            Error metric & Constructing $\varepsilon$-kernel & Rel.\ Nikodym distance $\varepsilon$ & Rel.\ error $\varepsilon$ \\\hline
            Deterministic & $\widetilde{\Theta}(\varepsilon^{-\frac{d-1}{2}})$ & $\widetilde{\Theta}(\varepsilon^{-\frac{d-1}{2}})$ & $\widetilde{\Theta}(\varepsilon^{-\frac{d-1}{2}})$ \\
            Randomized & $\widetilde{\Theta}(\varepsilon^{-\frac{d-1}{2}})$ & $\widetilde{\Theta}(\varepsilon^{-\frac{d-1}{2}})$ & $\widetilde{\Theta}(\varepsilon^{-\frac{2(d-1)}{d+3}})$ \\
            Quantum & $\widetilde{\Theta}(\varepsilon^{-\frac{d-1}{2}})$ & $\widetilde{\Theta}(\varepsilon^{-\frac{d-1}{2}})$ & $\widetilde{\Theta}(\varepsilon^{-\frac{d-1}{d+1}})$ \\
        \end{tabular}
        \caption{Overview of the query complexity results obtained in this paper for the convex set estimation and volume estimation problems. We only track the dependence on $1/\varepsilon$ and $R$, the prefactors are allowed to depend exponentially on the dimension $d$. The tilde hides polylogarithmic factors in $1/\varepsilon$ and $R$.}
        \label{tbl:results}
    \end{table}

    For the convex set estimation problem, we consider both the problem of constructing an $\varepsilon$-kernel, as well as estimating $K$ up to relative Nikodym distance $\varepsilon$. We obtain that the query complexities of both problems are $\widetilde{\Theta}(\varepsilon^{-(d-1)/2})$ in all three computational models. This shows in particular that randomness and quantumness alike do not provide any benefit over the deterministic model. Furthermore, for estimating a convex set up to relative Nikodym error, having access to a membership oracle is strictly more powerful than uniform sampling from the convex body, as we beat the $\Theta(\varepsilon^{-(d+1)/2})$ samples that are required in that setting~\cite{brunel2016adaptive}.

    For the volume estimation problem, we plot the obtained complexities in \cref{fig:complexities-graph}. For any fixed dimension $d$, we beat the previously best-known algorithms in the randomized and quantum settings, that respectively make $O(1/\varepsilon^2)$ and $O(1/\varepsilon)$ queries. The exponents in our complexities converge to these naive bounds as $d$ increases, showing that the obtained advantage becomes ever smaller with increasing $d$. Moreover, the gap between the randomized and quantum query complexity is small when $d$ is small, and becomes bigger as $d$ increases, converging to a full quadratic separation in the regime where $d$ is large. Like in the uniform sampling setting, the volume estimation problem is significantly easier than the convex set estimation problem in the randomized model.

    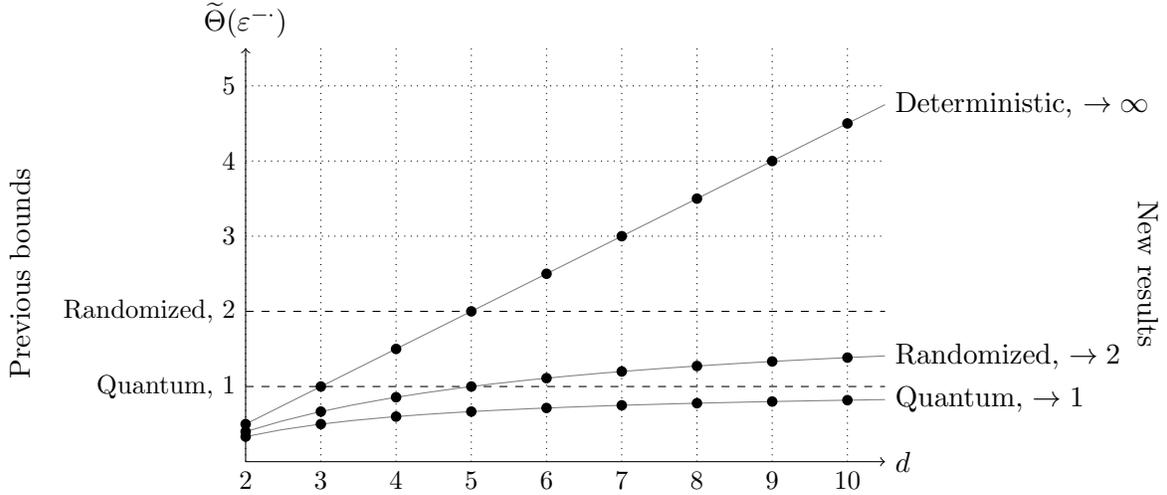
\begin{figure}[!ht]
        \centering
        \begin{tikzpicture}
            \draw[->] (2,0) to (2,5.5) node[above] {$\widetilde{\Theta}(\varepsilon^{-\cdot})$};
            \draw[->] (2,0) to (10.5,0) node[right] {$d$};
            \draw[domain=2:10.5, samples=19, variable=\x, gray] plot ({\x},{(\x-1)/2}) node[right] {\color{black}Deterministic, $\rightarrow \infty$};
            \draw[domain=2:10.5, samples=19, variable=\x, gray] plot ({\x},{2*(\x-1)/(\x+3)}) node[right] {\color{black}Randomized, $\rightarrow 2$};
            \draw[domain=2:10.5, samples=19, variable=\x, gray] plot ({\x},{(\x-1)/(\x+1)}) node[right] {\color{black}Quantum, $\rightarrow 1$};
            \foreach \x in {2,...,10} {
                \draw[dotted] (\x,0) node[below] {\small$\x$} to (\x,5.5);
                \fill (\x,{(\x-1)/2}) circle[radius=.07];
                \fill (\x,{2*(\x-1)/(\x+3)}) circle[radius=.07];
                \fill (\x,{(\x-1)/(\x+1)}) circle[radius=.07];
            }
            \foreach \y in {3,...,5} {
                \draw[dotted] (2,\y) node[left] {\small$\y$} to (10.5,\y);
            }
            \draw[dashed] (2,1) node[left] {\small Quantum, 1} to (10.5,1);
            \draw[dashed] (2,2) node[left] {\small Randomized, 2} to (10.5,2);
            \node[rotate=90] at (-1,2.5) {Previous bounds};
            \node[rotate=-90] at (14,2.5) {New results};
        \end{tikzpicture}
        \caption{Graph of the exponents in the query complexities for the volume estimation problem. $d$ is fixed, and the asymptotic limit is for $R \to \infty$ and $\varepsilon \downarrow 0$. The tilde hides polylogarithmic factors in $1/\varepsilon$ and $R$. The dashed lines represent the previously best-known results, and the solid ones connect the newly-found complexities.}
        \label{fig:complexities-graph}
    \end{figure}

    \paragraph{Techniques.}

    Our algorithms are based on the approaches taken in \cite{agarwal2004approximating,chan2006faster,yu2008practical}, and they all proceed in a similar manner. First, we apply an existing, deterministic rounding procedure~\cite[Theorem~4.6.1]{grotschel2012geometric} that makes $\widetilde{O}(1)$ queries and transforms the convex body into a ``well-rounded'' one, i.e., satisfying $B_d \subseteq K \subseteq R B_d$ with $R \in O(1)$ when $d$ is fixed.

    Then, we follow \cite{yu2008practical}, and take a set $\{v_j\}_{j=1}^n$ of $n$ roughly equally spaced points on the boundary of $(R+1)B_d$, with $n \in \Theta(\varepsilon^{-(d-1)/2})$. Every point $v_j$ is subsequently projected onto the convex body, and the convex hull of the resulting points forms an $\varepsilon$-kernel of $K$, as was shown in \cite{agarwal2004approximating}. Since these projections cannot be implemented directly in our model, we formulate every projection operation as a convex optimization problem, and use existing solvers from \cite{grotschel2012geometric} to obtain an approximate projection. It was already shown in \cite{chan2006faster} and \cite[Theorem~1]{yu2008practical} that this suffices to obtain an $\varepsilon$-kernel $\widetilde{K}$ of $K$. It also follows that $\widetilde{K}$ is an $\varepsilon$-precise approximation of $K$ in the relative Nikodym metric.

    Next, we find a speed-up of this approach for the volume estimation problem. To that end, we first run the convex set estimation procedure to some precision $\varepsilon' > \varepsilon$, to obtain an $\varepsilon'$-kernel $\underline{K}$ of $K$. We slightly enlarge the inner approximation $\underline{K} \subseteq K$ to generate an outer approximation $\overline{K} \supseteq K$ of $K$, with the property that $\Vol(\overline{K} \setminus \underline{K})/\Vol(K) \in O(\varepsilon')$. Then, in the randomized setting, we sample $\Theta((\varepsilon'/\varepsilon)^2)$ times uniformly from the region $\overline{K} \setminus \underline{K}$, and compute the fraction of samples $\xi$ that is inside $K$. We then use this as a refinement to our estimate of $\Vol(K)$, i.e., we output $\Vol(\underline{K}) + \xi\Vol(\overline{K} \setminus \underline{K})$. In the quantum case, we employ the same idea, but this time we use quantum amplitude estimation to quadratically speed up the computation of $\xi$. In both cases, balancing the number of queries made in the first and second step yields the complexity bounds from \cref{tbl:results}.

    The matching lower bounds for the query complexities follow from a construction that dates back to Dudley~\cite{dudley1974metric}, where it was used to prove a lower bound on the number of facets needed for approximating convex bodies by polytopes. The construction packs $n$ disjoint, identical spherical caps on the boundary of the unit ball through a $\delta$-net, for some suitable value of $\delta$ that depends on $d$ and $\eps$. Then, a convex body $K_x \subseteq B_d$ is associated to a bitstring $x$ of length $n$ by including the $j$th spherical cap in $K_x$ if and only if $x_j = 1$. See \cref{fig:lb-construction} for an illustration.

    The core observation is that estimating $K_x$, or its volume, becomes equivalent to learning a fraction of the entries of the bitstring $x$, or its Hamming weight, respectively. Additionally, a membership query to $K_x$ can be simulated with a single query to $x$. Recovering a constant fraction of bits of an $n$-bit string is known to cost $\Theta(n)$ queries in all three models that we consider. Similarly, the problem of estimating the Hamming weight of a bitstring is also known as the approximate counting problem and its (deterministic, randomized and quantum) query complexity has been tightly characterized in earlier work. We show that by carefully choosing the size of the spherical caps, this construction provides matching lower bounds for the query complexities from \cref{tbl:results} up to polylogarithmic factors in $R$ and $\varepsilon$ (and large factors in the dimension~$d$).

    \paragraph{History in the high-dimensional setting.}
    For completeness, we also comment on volume estimation in the high-dimensional setting, which tracks the dependency of the query complexity on the dimension $d$. The title of this work is a reference to the early survey by Simonovits called ``How to compute the volume in high dimension?''~\cite{simonovits2003compute}, and we refer the interested reader to this survey for a more in-depth discussion of this setting.
    The main takeaways in this setting are the following:
    \begin{itemize}
        \item[(i)]
        In the \emph{deterministic} setting, any deterministic algorithm must make a number of queries that is exponential in~$d$ \cite{elekes1986geometric,barany1987computing}.
        \item[(ii)]
        In the \emph{randomized} setting, there famously exists a Markov chain Monte Carlo algorithm that makes a number of queries polynomially in $d$ and $1/\varepsilon$ \cite{dyer1991random}.
        While the initial bound scaled as $\widetilde{O}(d^{23}/\varepsilon^2)$, a long of research led to the current best bound that is $\widetilde{O}(d^3/\varepsilon^2)$.
        This bound follows from combining algorithms in \cite{CV18,jia2021reducing} with recent breakthroughs on the so-called KLS-constant~(see e.g.~\cite{klartag2022bourgain}), and it approaches a lower bound of $\widetilde{\Omega}(d^2)$ on the randomized query complexity \cite{rademacher2008dispersion}.
        \item[(iii)] In the \emph{quantum} setting, Chakrabarti, Childs, Hung, Li, Wang and Wu \cite{chakrabarti2023quantum} gave an improved bound of the form $\widetilde{O}(d^3 + d^{2.5}/\varepsilon)$.\footnote{In \cite[Page~20:11]{chakrabarti2023quantum}, the authors claim that recent breakthroughs in the KLS-conjecture~\cite{chen2021almost,jambulapati2022slightly,klartag2023logarithmic} improve the analysis of their volume estimation algorithm to $\widetilde{O}(d^{2.5+o(1)} + d^{2+o(1)}/\varepsilon)$ membership oracle queries. They argue that the mixing time of the hit-and-run walk is $\widetilde{O}(d^2\psi^2)$, where $\psi$ is the KLS-constant. This indeed appears to be a correct upper bound on the mixing time of the \textit{ball walk}, see \cite[Theorem~2.7]{jia2021reducing}, but, as far as we are aware, the mixing time of the hit-and-run walk is not known to be improved due to the resolution of the KLS-conjecture. Adapting the algorithm of \cite{chakrabarti2023quantum} to make use of the ball walk is highly non-trivial. Consequently, the question of whether volume estimation can be done in $\widetilde{O}(d^{2.5+o(1)} + d^{2+o(1)}/\varepsilon)$ quantum queries to the membership oracle is still open.} This was later improved to $\widetilde{O}(d^3 + d^{2.25}/\varepsilon)$ by Cornelissen and Hamoudi \cite{cornelissen2023sublinear}. Both use quantum algorithms to speed up the Markov chain Monte Carlo method.
        The authors of \cite{chakrabarti2023quantum} also prove that the query complexity is $\Omega(\sqrt{d})$ when $\varepsilon$ is constant, and $\Omega(1/\varepsilon)$ in the regime where $1/d \leq \varepsilon \leq 1/3$, which in turn implies $\Omega(d)$ when $\varepsilon = 1/d$.\footnote{In \cite{chakrabarti2023quantum}, the authors additionally informally suggest a lower bound of $\Omega(1/\varepsilon)$ in general, but the corresponding theorem statement excludes the limit $\varepsilon \downarrow 0$. A lower bound of $\Omega(1/\varepsilon)$ for fixed $d$ and in the limit where $\varepsilon \downarrow 0$ would indeed contradict our results.}
    \end{itemize}

    \paragraph{Organization.}

    In \cref{sec:preliminaries}, we fix notation, formally define algorithmic models and computational problems, and state results from (computational) geometry. Subsequently, in \cref{sec:algorithms} we develop the algorithms, and in \cref{sec:lower-bounds} we prove the corresponding lower bounds.

    \section{Preliminaries}
    \label{sec:preliminaries}

    \subsection{Notation}

    We start by fixing some notation. Let $\N$ be the set of all positive integers. For all $n \in \N$, let $[n] = \{1, \dots, n\}$, and $[n]_0 = \{0\} \cup [n]$. For any $k \in \N$ (or $x \in \R$), we let $\N_{\geq k}$ (resp.\ $\R_{\geq x}$) be the set of all integers (resp.\ reals) that are bigger than or equal to $k$ (resp.~$x$). We denote the Euclidean norm by $\norm{\cdot}$. For every $d \in \N$, we denote the unit ball in $d$ dimensions by $B_d = \{x \in \R^d : \norm{x} \leq 1\}$. We write $\Gamma_d = \Vol(B_d)$ for the volume of $B_d$. We denote the unit sphere in $d$ dimensions by $S_{d-1} = \partial B_d$. For two sets  $A$ and $B$ we let $A \Delta B = (A \setminus B) \cup (B \setminus A)$ be the symmetric difference of $A$ and $B$.

    A set $K \subseteq \R^d$ is convex if for any two points $x,y \in K$, the straight line segment between $x$ and $y$ is fully contained in $K$, i.e., for every $\lambda \in [0,1]$, $\lambda x + (1-\lambda)y \in K$. For two convex sets $K,K' \subseteq \R^d$, we recall that $K \cap K'$ and $K + K' = \{x+y : x \in K, y \in K'\}$ are again convex. For ease of notation, we typically make an implicit assumption that all convex bodies are closed.

    We use big-$O$-notation to hide constant factors. That is, for two functions $f,g : \R_{\geq 0}^d \to \R_{\geq 0}$, we write $f \in O(g)$ if there exist $C,r > 0$ such that $f(x) \leq Cg(x)$ whenever $\norm{x} \geq r$. We write $f \in \Omega(g)$ if $g \in O(f)$, and we write $f \in \Theta(g)$ if $f \in O(g) \cup \Omega(g)$. Sometimes, the limit in which the big-$O$-notation is to be interpreted plays an important role. For instance if the big-$O$-notation holds in the limit where $\varepsilon \downarrow 0$, then we interpret the function $f,g$ as functions in the variable $1/\varepsilon$.

    Similarly, we use the big-$O$-tilde-notation to hide polylogarithmic factors. That is, if $f,g : \R_{\geq 0}^d \to \R_{\geq 0}$, then we write $f \in \widetilde{O}(g)$ if there exist constants $C,k,r \geq 0$ such that $f(x) \leq Cg(x)\prod_{j=1}^d\log^k(x_j)$ whenever $\norm{x} \geq r$. We write $f \in \widetilde{\Omega}(g)$ if $g \in \widetilde{O}(f)$, and $f \in \widetilde{\Theta}(g)$ if $f \in \widetilde{O}(g) \cup \widetilde{\Omega}(g)$.

    \subsection{Access model and computational models}

    The input to the algorithms we design in this work are a dimension $d \in \N$, outer radius $R \geq 1$, precision $\varepsilon \in (0,1)$, and a convex body $K \subseteq \R^d$. The \textit{access model} specifies how the algorithm can access this input. Throughout, we will assume that the parameters $d$, $R$ and $\varepsilon$ are set beforehand, and hence are known to the algorithm. As for the convex body $K$, many ways to access it have been considered in the literature. We refer to \cite[Chapter~2]{grotschel2012geometric} for an overview of several different access models. In this work, we assume that it can be accessed by means of a (strong) \textit{membership oracle}, in line with \cite[Definition~2.1.5]{grotschel2012geometric}.

    \begin{definition}[Membership oracle]
        \label{def:membership-oracle}
        Let $d \in \N$, and $K \subseteq \R^d$ be a convex body. A \textit{membership oracle to $K$} is a procedure that, given a point $x \in \R^d$, outputs whether $x \in K$.
    \end{definition}

    In general, an \textit{oracle} is a subroutine that the algorithm can run only as a black box, i.e., it has no more information about what happens when it is run. One such call to the oracle is referred to as a query. The \textit{query complexity} of a problem is the minimum number of queries any algorithm needs to make to solve said problem.

    Additional to the access model, we also specify \textit{computational models} that define the operations that the algorithms are allowed to use. We consider three different computational models in this work. In the \textit{deterministic model}, the algorithm is completely described beforehand and follows a deterministic computational path. In the \textit{randomized model}, the algorithm is allowed to use randomness during the execution of the algorithm, and in particular let the oracle's inputs depend on it. It is required to output a correct answer with probability at least $2/3$. In the \textit{quantum model}, the algorithm is additionally allowed to perform operations on a quantum state, and supply the oracle's inputs in a superposition, which it answers coherently. For a more detailed introduction to quantum algorithms, we refer to~\cite{nielsen2010quantum}.

    Throughout the execution of our algorithms, we additionally assume that we can store vectors in $\R^d$ precisely in memory in all three models, and that we can perform basic arithmetic with them exactly. We also assume in the randomized model that we can sample from any distribution on $\R^d$ as long as we have a classical description of it, and we use this to sample uniformly from an arbitrary region in \cref{thm:randomized-volume-estimation}. Similarly, in the quantum model, we assume that we can generate a uniform superposition over an arbitrary region in $\R^d$, and we use this in \cref{thm:quantum-volume-estimation}. Both of these operations do not depend on $K$, so in an implementation of the algorithms presented in this work, one can suppress the finite-precision errors that arise from these assumptions arbitrarily without increasing the number of membership oracle queries.

    \subsection{Rounding convex bodies}
    \label{subsec:rounding}

    We start by describing a routine that ``rounds'' the convex body. A similar idea appears in \cite{agarwal2004approximating,chan2006faster,yu2008practical}, but in contrast to their approach, we use inner and outer ellipses, rather than inner and outer cubes in our rounding procedure. This difference is not fundamental, it merely allows us to use a rounding routine that already exists in the literature~\cite{grotschel2012geometric}.

    The aim of ``rounding'' a convex body $K \subseteq \R^d$ is to find an invertible affine linear map $L : x \mapsto x_0 + Tx$, such that $B_d \subseteq L(K) \subseteq R'B_d$ with $R' \geq 1$ as small as possible. Intuitively, one can think of rounding as finding a linear transformation that compresses the convex body as much as possible, so that it does not have parts that stick out far from the origin. Most of the existing literature focuses on randomized algorithms to compute this affine linear map~\cite{lovasz2006simulated,jia2021reducing}, with the current state-of-the-art obtaining $R' \in \widetilde{O}(\sqrt{d})$ and requiring $\widetilde{O}(d^3)$ queries. In our setting, though, we require a deterministic procedure, and since we take $d$ to be a constant we don't necessarily need $R'$ to scale optimally in the dimension. Below, we sketch how to obtain a deterministic rounding procedure that obtains $R' \in O(d^3)$ and runs in time $O(\poly(d,\log(R)))$.

    \begin{theorem}[{\cite{grotschel2012geometric}}]
        \label{thm:rounding}
        Let $d \in \N$, $R \geq 1$, and let $K \subseteq \R^d$ be convex such that $B_d \subseteq K \subseteq RB_d$. Then, there is a deterministic algorithm that makes $O(\poly(d,\log(R))$ queries, and finds an invertible affine linear map $L$ such that $B_d \subseteq L(K) \subseteq R'B_d$, with $R' = d(d+1)^2 \in O(d^3)$. When $d$ is fixed, the algorithm makes $\widetilde{O}(1)$ queries, and $R' \in O(1)$.
    \end{theorem}

    \begin{proof}
        Since $K \subseteq \R^d$ satisfies $B_d \subseteq K \subseteq RB_d$, we observe from \cite[Figure~4.1]{grotschel2012geometric} that we can deterministically turn a membership oracle into a weak separation oracle. The total multiplicative factor incurred from this conversion is $O(\poly(d,\log(R)))$, and hence is $\widetilde{O}(1)$ in the case where $d$ is fixed.

        Subsequently, from \cite[Theorem~4.6.1]{grotschel2012geometric}, we find that with $O(\poly(d,\log(R)))$ calls to a weak separation oracle to $K$, we can find an ellipsoid $E$, such that $E \subseteq K \subseteq d(d+1)^2E$. Thus, by letting $L$ be the affine linear transformation that maps $E$ to the unit ball, we obtain that
        \[B_d \subseteq L(K) \subseteq R'B_d, \qquad \text{where} \qquad R' = d(d+1)^2 \in O(d^3).\]
        For fixed $d$, we obtain $R' \in O(1)$, and the query complexity of finding $L$ is in $\widetilde{O}(1)$.
    \end{proof}

    \subsection{Geometry}

    We proceed by recalling some theoretical background in geometry. We start by formally introducing the concept of a net.

    \begin{definition}[$\delta$-net]
        Let $d \in \N$, $\delta > 0$, and let $S \subseteq \R^d$ be any set. We say that $N \subseteq S$ is a $\delta$-net of $S$ if
        \begin{enumerate}[nosep]
            \item For any two distinct points $v,w \in N$, we have $\norm{v-w} \geq \delta$.
            \item For any $v \in S$, there exists a $w \in N$ such that $\norm{v - w} \leq \delta$.
        \end{enumerate}
    \end{definition}

    Intuitively, one can think of a $\delta$-net $N$ as the centers of a set of balls of radius $\delta$, that together cover the set $S$. Moreover, since these centers are at least separated by distance $\delta$, the balls with radius $\delta/2$ centered at the points in $N$ must be disjoint up to a measure-zero set. By comparing the surface area of the sphere to that of $\{u \in S_{d-1}: \|u-v\| \leq \delta\}$ (for some fixed $v \in S_{d-1}$), the following bound follows on the number of points in a $\delta$-net on the sphere. We stress that the assumption that $d$ is fixed is crucial in the above proposition. Indeed, the $\Theta$-notation hides a prefactor that might depend exponentially on $d$.

    \begin{proposition}
        \label{lem:num-patches-lb}
        Let $d \in \N$ be fixed. Let $\delta \in (0,1)$, and let $N_{\delta}$ be a $\delta$-net of $S_{d-1}$. Then, $|N_{\delta}| \in \Theta(\delta^{-(d-1)})$.
    \end{proposition}

    \begin{proof}
        For any $v \in V$ and $r > 0$, let $R_{v,r} = \{w \in S_{d-1} : \norm{v-w} \leq r\}$. The surface areas of $R_{v,r}$ is $A_{v,r} \in \Theta(r^{d-1})$ and the surface area of $S_{d-1}$ is $A_{d-1} \in \Theta(1)$. Since $N_{\delta}$ is a $\delta$-net, all $R_{v,\delta/2}$'s with $v \in N_{\delta}$ are disjoint up to possibly some measure-zero set, and all $R_{v,\delta}$'s cover $S_{d-1}$. Thus, we find $|N_{\delta}|A_{v,\delta/2} \leq A_{d-1} \leq |N_{\delta}|A_{v,\delta}$, from which we conclude that $|N_{\delta}| \in \Theta(\delta^{-(d-1)})$.
    \end{proof}

    We proceed with formally defining an $\varepsilon$-kernel.

    \begin{definition}[$\varepsilon$-kernel~{\cite{agarwal2004approximating}}]
        \label{def:eps-kernel}
        Let $d \in \N$, and $\varepsilon \in (0,1)$. Let $K \subseteq \R^d$ be a convex body. Then a convex body $\widetilde{K} \subseteq K$ is an $\varepsilon$-kernel of $K$ if
        \[\forall u \in S_{d-1}, \qquad \max_{x \in \widetilde{K}} u^Tx - \min_{x \in \widetilde{K}} u^Tx \geq (1-\varepsilon)\left[\max_{x \in K} u^Tx - \min_{x \in K} u^Tx\right].\]
    \end{definition}

    One very desirable property of $\varepsilon$-kernels is its invariance under affine linear transformations.

    \begin{lemma}[{\cite{agarwal2004approximating},\cite[Lemma~2.5]{agarwal2024computing}}]
        \label{prop:eps-kernel-affine-transform}
        Let $d \in \N$, $\varepsilon \in (0,1)$, and let $L$ be an invertible affine linear map on $\R^d$. Let $\widetilde{K},K \subseteq \R^d$ be convex sets. Then, $\widetilde{K}$ is an $\varepsilon$-kernel of $K$ if and only if $L(\widetilde{K})$ is an $\varepsilon$-kernel of $L(K)$.
    \end{lemma}

    Next, we make a connection between $\varepsilon$-kernels and approximations w.r.t.\ the Hausdorff metric.

    \begin{lemma}
        \label{lem:eps-kernel-hausdorff}
        Let $d \in \N$, $R \geq 1$, $\varepsilon > 0$, and let $K, \widetilde{K} \subseteq \R^d$ be convex sets.
        \begin{enumerate}[nosep]
            \item If $B_d \subseteq K$ and $\widetilde{K} \subseteq K$ is an $\varepsilon$-precise Hausdorff approximation of $K$, then $\widetilde{K}$ is also an $\varepsilon$-kernel of $K$.
            \item If $K \subseteq RB_d$ and $\widetilde{K}$ is an $\varepsilon$-kernel of $K$, it is a $2\varepsilon R$-precise Hausdorff approximation of $K$.
        \end{enumerate}
    \end{lemma}

    \begin{proof}
        Observe that if $\widetilde{K} \subseteq K$, then the Hausdorff distance $d$ between $\widetilde{K}$ and $K$ is
        \begin{equation}
            \label{eq:hausdorff}
            d = \max_{u \in S_{d-1}} \left[\max_{x \in K} u^Tx - \max_{x \in \widetilde{K}} u^Tx\right] = \max_{u \in S_{d-1}} \left[\min_{x \in \widetilde{K}} u^Tx - \min_{x \in K} u^Tx\right].
        \end{equation}
        For the first claim, we take $u \in S_{d-1}$ arbitrarily, and observe from \cref{eq:hausdorff} that
        \[\max_{x \in \widetilde{K}} u^Tx - \min_{x \in \widetilde{K}} u^Tx \geq \max_{x \in K} u^Tx - \min_{x \in K} u^Tx - 2\varepsilon \geq (1-\varepsilon)\left[\max_{x \in K} u^Tx - \min_{x \in K} u^Tx\right].\]
        For the second claim, let $d$ be the Hausdorff distance between $K$ and $\widetilde{K}$, and $u$ the vector that maximizes the middle expression in \cref{eq:hausdorff}. Then, we have
        \begin{align*}
            d &= \max_{x \in K} u^Tx - \max_{x \in \widetilde{K}} u^Tx \leq \left[\max_{x \in K} u^Tx - \min_{x \in K} u^Tx\right] - \left[\max_{x \in \widetilde{K}} u^Tx - \min_{x \in \widetilde{K}} u^Tx\right] \\
            &\leq \varepsilon\left[\max_{x \in K} u^Tx - \min_{x \in K} u^Tx\right] \leq \varepsilon \cdot \diam(K) \leq 2\varepsilon R.\qedhere
        \end{align*}
    \end{proof}

    Finally, we observe that an $\varepsilon$-kernel naturally provides a relative Nikodym approximation of $K$, and can be used to construct an outer approximation of $K$, in the following proposition.

    \begin{proposition}
        \label{prop:eps-kernel-volumetric-approximation}
        Let $d \in \N$, $\varepsilon \in (0,1/(4d(d+1)^2))$, and let $\widetilde{K}$ be an $\varepsilon$-kernel of a convex set $K \subseteq \R^d$. Then, given a full description of $\widetilde{K}$, we can construct a convex set $\overline{K} \supseteq K$ such that $\Vol(\overline{K} \setminus \widetilde{K})/\Vol(K) \leq (1+4\varepsilon d(d+1)^2)^d - 1 \in O(\varepsilon)$, with $\widetilde{O}(1)$ membership oracles to $K$. In particular, this implies that $\Vol(\widetilde{K} \Delta K)/\Vol(K) \in O(\varepsilon)$.
    \end{proposition}

    \begin{proof}
        Since being an $\varepsilon$-kernel and being a relative Nikodym-distance approximation are both invariant under affine linear transformations, we can without loss of generality first round the convex body $K$ using \cref{thm:rounding}, using just $\widetilde{O}(1)$ membership queries. It then suffices to consider the case where $B_d \subseteq K \subseteq RB_d$, with $R = d(d+1)^2 \in O(1)$. Since $\varepsilon < 1/(4R)$, we find by \cref{lem:eps-kernel-hausdorff} that $B_d \subseteq K \subseteq \widetilde{K} + 2\varepsilon RB_d \subseteq \widetilde{K} + B_d/2$, and so $B_d/2 \subseteq \widetilde{K}$. Thus,
        \[B_d/2 \subseteq \widetilde{K} \subseteq K \subseteq \widetilde{K} + 2\varepsilon RB_d \subseteq (1+4\varepsilon R)\widetilde{K} =: \overline{K},\]
        from which we find that
        \[\frac{\Vol(\overline{K} \setminus \widetilde{K})}{\Vol(K)} = \frac{\Vol(\overline{K}) - \Vol(\widetilde{K})}{\Vol(K)} = \frac{[(1+4\varepsilon R)^d - 1]\Vol(\widetilde{K})}{\Vol(K)} \leq (1+4\varepsilon R)^d - 1 \in O(\varepsilon).\qedhere\]
    \end{proof}

    \subsection{Problem definitions}

    We consider three computational problems in this paper. We formally introduce them here.

    \begin{definition}[Problem definitions]
        \label{def:problems}
        Let $d \in \N$, $R > 1$, and $\varepsilon \in (0,1)$. Let $K \subseteq \R^d$ be a convex body such that $B_d \subseteq K \subseteq RB_d$. We consider three problems:
        \begin{enumerate}[nosep]
            \item The \emph{$\varepsilon$-kernel construction problem} is the problem of outputting an $\varepsilon$-kernel of $K$.
            \item The \emph{$\varepsilon$-Nikodym construction problem} is the problem of outputting a convex body $\widetilde{K} \subseteq \R^d$ such that $\Vol(K \Delta \widetilde{K}) \leq \varepsilon\Vol(K)$.
            \item The \emph{volume estimation problem} is the problem of outputting a non-negative real $\widetilde{V} \geq 0$ such that $|\Vol(K) - \widetilde{V}| \leq \varepsilon\Vol(K)$.
        \end{enumerate}
    \end{definition}

    It follows directly that these problems are qualitatively decreasing in terms of their query complexity. That is, solving the $\varepsilon$-kernel construction problem also solves the Nikodym construction problem with precision $O(\varepsilon)$, by virtue of \cref{prop:eps-kernel-volumetric-approximation}, which in turn solves the volume estimation problem with the same precision, since we can simply output the volume of the approximation $\widetilde{K}$. These relations are less clear if one considers the runtime instead, e.g., since computing the volume might be very costly.

    \section{Algorithms}
    \label{sec:algorithms}

    \subsection{Convex set estimation}
    \label{subsec:alg-conv-est}

    In this subsection, we develop a deterministic algorithm that constructs an $\varepsilon$-kernel of a well-rounded convex set $K \subseteq \R^d$, using membership queries to it. The algorithm follows the same general strategy as in \cite{yu2008practical}, and we replace their approximate nearest-neighbor queries by a query-efficient approximate projection onto a convex body. For this, we use the well-known observation that projection onto a convex body can be phrased as a convex optimization problem, which we can deterministically solve approximately with only polylogarithmically many membership oracles calls, using for example the ellipsoid method~\cite{grotschel2012geometric}.

    \begin{proposition}
        \label{prop:alg-proj}
        Let $d \in \N$, $R \geq 1$, $\varepsilon > 0$, $K \subseteq \R^d$ convex such that $B_d \subseteq K \subseteq RB_d$. Let $x \in \R^d \setminus K$, and $y'$ its projection onto $K$. There is a deterministic algorithm that obtains an element $\widetilde{y} \in K$ such that $\norm{\widetilde{y} - y'} \leq \varepsilon$. The algorithm makes $O(\poly(d,\log(\norm{x}),\log(R),\log(1/\varepsilon)))$ queries to a membership oracle of $K$. When $d$ is fixed and $R \in O(1)$, this complexity is $\widetilde{O}(1)$.
    \end{proposition}

    Now, we are ready to state the $\varepsilon$-kernel construction algorithm. We start by considering the well-rounded case, in \cref{alg:well-rounded-eps-kernel-construction}. We subsequently prove its properties in \cref{thm:well-rounded-eps-kernel-construction}.

    \begin{algorithm}[!ht]
        \caption{Well-rounded deterministic $\varepsilon$-kernel construction~\cite{yu2008practical}}
        \label{alg:well-rounded-eps-kernel-construction}
        \textbf{Input:}
        \begin{enumerate}[nosep]
            \item $d \in \N_{\geq 2}$: the dimension.
            \item $\varepsilon \in (0,1)$: the desired precision.
            \item $R \geq 1$: the outer radius, with $R \in O(1)$.
            \item $O_K$: a membership oracle that on input $x \in \R^d$ returns whether $x \in K$, where $B_d \subseteq K \subseteq RB_d$.
        \end{enumerate}
        \textbf{Derived constant:} $\eta = \sqrt{\varepsilon/R}$.

        \textbf{Output:} An $\varepsilon$-kernel $\widetilde{K}$ of $K$.

        \textbf{Number of queries:} $\widetilde{O}(\varepsilon^{-(d-1)/2})\quad$ (in the limit where $d$ is fixed).

        \textbf{Procedure:} $\texttt{RoundedEpsKernConstr}(d, \varepsilon, R, O_K)$:
        \begin{enumerate}[nosep]
            \item Let $\{x_j\}_{j=1}^n$ be an $\eta$-net of $(R+1)S_{d-1}$.
            \item For $j = 1, \dots, n$,
            \begin{enumerate}[nosep]
                \item Project $x_j$ onto $K$ with precision $\varepsilon$, using \cref{prop:alg-proj}. Denote the outcome by $p_j$.
            \end{enumerate}
            \item Output $\conv(\{p_j\}_{j=1}^n)$.
        \end{enumerate}
    \end{algorithm}

    \begin{theorem}
        \label{thm:well-rounded-eps-kernel-construction}
        Let $d \in \N_{\geq 2}$, $R \geq 1$ with $R \in O(1)$, and $K \subseteq \R^d$ be a convex body such that $B_d \subseteq K \subseteq RB_d$. Then, \cref{alg:well-rounded-eps-kernel-construction} computes an $\varepsilon$-kernel of $K$ that satisfies $\widetilde{K} \subseteq K$, with $\widetilde{O}(\varepsilon^{-(d-1)/2})$ membership oracle queries.
    \end{theorem}

    \begin{proof}
        Correctness is proven in \cite[Theorem~1]{yu2008practical}. For the bound on the number of queries, observe from \cref{lem:num-patches-lb} that the $\eta$-net contains $O(\eta^{-(d-1)}) = O(\varepsilon^{-(d-1)/2})$ points. Furthermore, for every point, we perform one approximate projection, which costs $\widetilde{O}(1)$ queries by \cref{prop:alg-proj}. As such, we conclude that the total number of membership oracle queries is $\widetilde{O}(\varepsilon^{-(d-1)/2})$.
    \end{proof}

    In the case where $K$ is not well-rounded, we combine \cref{alg:well-rounded-eps-kernel-construction} with the rounding procedure of \cref{thm:rounding}, and conversion to Nikodym distance in \cref{prop:eps-kernel-volumetric-approximation}, to obtain the following corollary.

    \begin{corollary}[Deterministic $\varepsilon$-kernel construction]
        \label{thm:eps-kernel-construction}
        Let $d \in \N_{\geq 2}$, $R \geq 1$, $\varepsilon \in (0,1/(4d(d+1)^2))$, and $K \subseteq \R^d$ be convex such that $B_d \subseteq K \subseteq RB_d$. Then, we can deterministically compute an $\varepsilon$-kernel $\widetilde{K}$ of $K$ with $\widetilde{O}(\varepsilon^{-(d-1)/2})$ membership oracle queries. $\widetilde{K}$ is also an $O(\varepsilon)$-precise approximation of $K$ in the relative Nikodym distance, and $\Vol(\widetilde{K})$ is an $O(\varepsilon)$-precise relative estimate of $\Vol(K)$.
    \end{corollary}

    This corollary proves all the $\widetilde{O}(\varepsilon^{-(d-1)/2})$ upper bounds in \cref{tbl:results}. In \cref{subsec:reductions}, we prove that this approach is essentially optimal for the reconstruction problem in all computational models, i.e., even in the randomized and quantum settings we cannot improve significantly over the approach taken in \cref{thm:eps-kernel-construction}.

    \subsection{Volume estimation}
    \label{subsec:alg-vol-est}

    In this subsection, we switch to the volume estimation problem. To start off, we remark that \cref{thm:eps-kernel-construction} solves it in the deterministic setting. In the randomized and quantum settings, however, we can obtain a significant improvement over this approach.

    The core idea is to run the (deterministic) $\varepsilon'$-kernel construction algorithm up to some worse precision $\varepsilon' > \varepsilon$, to obtain an inner approximation $\underline{K} \subseteq K$. We can the use \cref{prop:eps-kernel-volumetric-approximation} to obtain an outer approximation $\overline{K} \supseteq K$, such that $\Vol(\overline{K} \setminus \underline{K})/\Vol(K) \in O(\varepsilon')$.

    Subsequently, in the randomized setting, we sample uniformly from $\overline{K} \setminus \underline{K}$, and compute the fraction of points $\xi$ that is inside $K$. We then compute $\Vol(\underline{K}) + \xi\Vol(\overline{K} \setminus \underline{K})$ as a refined estimate of the volume of $K$, eventually resulting in the following theorem.

    \begin{theorem}[Randomized volume estimation]
        \label{thm:randomized-volume-estimation}
        Let $d \in \N_{\geq 2}$, $R \geq 1$, $\varepsilon \in (0,1)$, $K \subseteq \R^d$ convex such that and $B_d \subseteq K \subseteq RB_d$. We can compute $\widetilde{V} \geq 0$ such that $|\widetilde{V} - \Vol(K)| \leq \varepsilon\Vol(K)$ with probability at least $2/3$, using a randomized algorithm that makes $\widetilde{O}(\varepsilon^{-2(d-1)/(d+3)})$ queries to a membership oracle of $K$.
    \end{theorem}

    \begin{proof}
        We first use \cref{thm:eps-kernel-construction} to obtain an $\varepsilon'$-kernel $\underline{K} \subseteq K$ with
        \[\varepsilon' = \frac{(1+\varepsilon^{\frac{4}{d+3}})^{\frac1d}-1}{4d(d+1)^2} \in \Theta\left(\varepsilon^{\frac{4}{d+3}}\right).\]
        We make $\widetilde{O}((\varepsilon')^{-(d-1)/2}) = \widetilde{O}(\varepsilon^{-2(d-1)/(d+3)})$ membership oracles in this step of the algorithm. Then, we use \cref{prop:eps-kernel-volumetric-approximation} to generate an outer approximation $\overline{K} \supseteq K$, using just $\widetilde{O}(1)$ membership oracle queries. We find that $\Vol(\overline{K}/\underline{K})/\Vol(K) \leq \varepsilon'$.

        Next, we take $n := \lceil 3(\varepsilon'/\varepsilon)^2\rceil$ random samples from $\overline{K} \setminus \underline{K}$, and we check for each of them whether they are in $K$. We denote the fraction of them that is in $K$ by $\xi$. This step requires $n = \widetilde{O}((\varepsilon'/\varepsilon)^2) = \widetilde{O}(\varepsilon^{-2(d-1)/(d+3)})$ membership oracle queries too.

        Finally, we output $\widetilde{V} = \Vol(K) + \xi\Vol(\overline{K} \setminus \underline{K})$. We observe that
        \[\E[\widetilde{V}] = \Vol(\underline{K}) + \E[\xi]\Vol(\overline{K} \setminus \underline{K}) = \Vol(\underline{K}) + \frac{\Vol(K \setminus \underline{K})}{\Vol(\overline{K} \setminus \underline{K})}\Vol(\overline{K} \setminus \underline{K}) = \Vol(K),\]
        and
        \[\Var[\widetilde{V}] = \Vol(\overline{K} \setminus \underline{K})^2\Var[\xi] \leq \frac{(\varepsilon' \cdot \Vol(K))^2}{n} \leq \frac{(\varepsilon \cdot \Vol(K))^2}{3}.\]
        We conclude with Chebyshev's inequality that
        \[\P\left[\left|\widetilde{V} - \Vol(K)\right| > \varepsilon\Vol(K)\right] \leq \frac{\Var[\widetilde{V}]}{(\varepsilon \cdot \Vol(K))^2} \leq \frac13.\qedhere\]
    \end{proof}

    In the quantum case, we speed up the sampling phase from the randomized algorithm by replacing it with a quantum primitive known as amplitude estimation~\cite{BHMT98}. We attain a quadratic speed-up of this step, and the claimed complexity follows after rebalancing the costs of the different steps.

    \begin{theorem}[Quantum volume estimation]
        \label{thm:quantum-volume-estimation}
        Let $d \in \N_{\geq 2}$, $R \geq 1$, $\varepsilon \in (0,1)$, and $K \subseteq \R^d$ convex such that $B_d \subseteq K \subseteq RB_d$. We can compute $\widetilde{V} \geq 0$ such that $|\widetilde{V} - \Vol(K)| \leq \varepsilon\Vol(K)$ with probability at least $2/3$, using a quantum algorithm that makes $\widetilde{O}(\varepsilon^{-(d-1)/(d+1)})$ queries to $O_K$.
    \end{theorem}

    \begin{proof}
        Similarly to the proof of \cref{thm:randomized-volume-estimation}, we start by constructing an $\varepsilon'$-kernel $\underline{K} \subseteq K$ using \cref{thm:eps-kernel-construction}, and then use \cref{prop:eps-kernel-volumetric-approximation} to obtain $\overline{K} \supseteq K$ such that $\Vol(\overline{K} \setminus \underline{K}) \leq \varepsilon'\Vol(K)$. In contrast to \cref{thm:randomized-volume-estimation}, though, we choose
        \[\varepsilon' = \frac{(1 + \varepsilon^{\frac{2}{d+1}})^{\frac1d} - 1}{4d(d+1)^2} \in \Theta\left(\varepsilon^{\frac{2}{d+1}}\right),\]
        which brings the total number of membership oracle queries in this part to $\widetilde{O}((\varepsilon')^{-(d-1)/2}) = \widetilde{O}(\varepsilon^{-(d-1)/(d+1)})$.

        Next, we use amplitude estimation to find an estimate $\xi$ of $\Vol(K \setminus \underline{K}) / \Vol(\overline{K} \setminus \underline{K})$ up to precision $\varepsilon/\varepsilon'$. To that end, we compute the overlap between the uniform superposition on $\overline{K} \setminus \underline{K}$, and the subspace spanned all the quantum states representing vectors in $K$. From \cite{BHMT98}, we find that this can be done with success probability at least $8/\pi^2 > 2/3$, and with a total number of queries that satisfies $\widetilde{O}(\varepsilon'/\varepsilon) = \widetilde{O}(\varepsilon^{-(d-1)/(d+1)})$. Finally, we output $\widetilde{V} = \Vol(K) + \xi\Vol(\overline{K} \setminus \underline{K})$, and observe that
        \[|\widetilde{V} - \Vol(K)| = \Vol(\overline{K} \setminus \underline{K}) \cdot \left|\xi - \frac{\Vol(K \setminus \underline{K})}{\Vol(\overline{K} \setminus \underline{K})}\right| \leq \varepsilon' \cdot \Vol(K) \cdot \frac{\varepsilon}{\varepsilon'} = \varepsilon \cdot \Vol(K).\qedhere\]
    \end{proof}

    This concludes our discussion of randomized and quantum algorithms estimating the volume of a convex body. In \cref{subsec:reductions}, we prove that these algorithms are essentially optimal.

    \section{Lower bounds}
    \label{sec:lower-bounds}

    In this section, we prove matching lower bounds for the computational problems from \cref{def:problems}. We can always rescale our problems, and so we can freely change the assumption that $B_d \subseteq K \subseteq RB_d$ to $(1/R)B_d \subseteq K \subseteq B_d$. For ease of notation, we will phrase everything with the latter assumption in mind.

    The lower bounds crucially rely on embedding a bitstring on the boundary of the unit ball in $d$ dimensions. To that end, we let $\{v_j\}_{j=1}^n$ be a net of $S_{d-1}$, and for every $j \in [n]$, we define a \textit{spherical cap} $P_j$ around $v_j$, i.e., a small region cut off from the unit ball by a hyperplane that is orthogonal to $v_j$. We take all the spherical caps to be disjoint and with equal volume, and we let $K_0$ be the remaining part of the unit ball, as shown in \cref{fig:lb-construction}.

    \begin{figure}[!ht]
        \centering
        \begin{tikzpicture}[scale=1.5]
            \def\num{6}
            \fill (0,0) circle[radius=.035];
            \node at (.5,0) {$K_0$};
            \foreach \j in {1,...,\num} {
                \draw[fill=gray!20] ({1/cos(180/\num)*cos(360*\j/\num)},{1/cos(180/\num)*sin(360*\j/\num)}) arc({360/\num*\j}:{360/\num*(\j-1)}:{1/cos(180/\num)}) to cycle;
                \draw[->] (0,0) to ({1/cos(180/\num)*cos(360*(\j-.5)/\num)},{1/cos(180/\num)*sin(360*(\j-.5)/\num)});
                \node at ({(1/cos(180/\num)+.12)*cos(360*(\j-.5)/\num)},{(1/cos(180/\num)+.12)*sin(360*(\j-.5)/\num)}) {$v_{\j}$};
                \node at ({(1/cos(180/\num)-.08)*cos(360*(\j-.4)/\num)},{(1/cos(180/\num)-.08)*sin(360*(\j-.4)/\num)}) {\tiny$P_{\j}$};
            }
        \end{tikzpicture}
        \caption{The lower bound construction in two dimensions with $n = 6$. For any $x \in \{0,1\}^6$, the convex body $K_x$ is formed by taking the union of $K_0$ and all spherical caps $P_j$ if and only if the corresponding bit $x_j$ is $1$.}
        \label{fig:lb-construction}
    \end{figure}
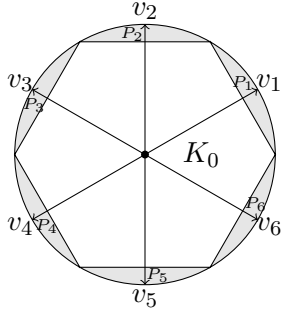

    The core idea is to associate a convex body $K_x$ to every bitstring $x \in \{0,1\}^n$, defined as the union of $K_0$, and all the spherical caps $P_j$ if and only if $x_j = 1$. We then immediately observe that $\Vol(K_x) = \Vol(K_0) + |x|\Vol(P_1)$, and $\Vol(K_x \Delta K_y) = |x \oplus y| \Vol(P_1)$.

    Intuitively, now, if we find an approximation to the convex body $K_x$, then we also find an approximation to the bitstring $x$. Similarly, if we find a sufficiently precise estimate of $K_x$'s volume, we also obtain a good approximation of $x$'s Hamming weight. In short, we can reduce the bitstring recovery problem to convex set estimation, and the Hamming weight estimation problem to volume estimation, and hence lower bounds on the former imply lower bounds on the latter.

    \subsection{Query complexity of bitstring problems}
    \label{subsec:bit-string-problems}

    We first recall the bitstring recovery and the Hamming weight estimation problems. These are well-studied in the existing literature, albeit a bit scattered, and we gather the existing results in two concise theorem statements. The proofs can be found in \cref{app:bit-string-problems}.

    \begin{restatable}[Bitstring recovery problem]{theorem}{bitstringrecovery}
        \label{thm:bit-string-recovery}
        Let $n \in \N$, and let $x \in \{0,1\}^n$ be a bitstring that we can access through bit queries. Suppose we wish to output a bitstring $y \in \{0,1\}^n$ such that $|x \oplus y| \leq n/4$. The query complexities for this problem are $\Theta(n)$ in the deterministic, randomized and quantum setting.
    \end{restatable}

    Next, we consider the Hamming weight estimation problem. This problem is sometimes also referred to as the \textit{approximate counting problem}. In short, given a bitstring of length $n$, we wish to estimate the number of $1$'s it contains up to some additive precision $k$. Its hardness is also well-understood -- we gather the previous known results in \cref{thm:hamming-weight-estimation}.

    \begin{restatable}[Hamming weight estimation problem]{theorem}{hammingweightestimation}
        \label{thm:hamming-weight-estimation}
        Let $n \in \N$ and let $x \in \{0,1\}^n$ be a bitstring that we can access through bit queries. Let $k \in \N$ be such that $1 \leq k \leq n/4$. Suppose we wish to output $w \in [n]_0$ such that $||x| - w| \leq k$. The query complexities for this problem are $\Theta(n)$ in the deterministic setting, $\Theta(\min\{(n/k)^2,n\})$ in the randomized setting, and $\Theta(n/k)$ in the quantum setting.
    \end{restatable}

    \subsection{Reduction to convex set estimation and volume estimation}
    \label{subsec:reductions}

    We start by formalizing the concept of a \textit{spherical cap}, and provide a visualization in \cref{fig:delta-patch}.

    \begin{definition}[Spherical cap]
        Let $d \in \N$, $v \in S_{d-1}$ and $r > 0$. Let $S_v = \{u \in S_{d-1} : \norm{u-v} < r\}$, and let $P_v$ be its convex hull. Then $P_v$ is the spherical cap around $v$ with radius $r$.
    \end{definition}

    \begin{figure}[!ht]
        \centering
        \begin{tikzpicture}[scale=2]
            \begin{scope}[rotate=90]
                \draw ({cos(60)},{-sin(60)}) arc(-60:-45:1);
                \draw ({cos(60)},{sin(60)}) arc(60:45:1);
                \draw[thick] ({cos(45)},{-sin(45)}) arc(-45:45:1);
                \node at ({1.15*cos(25)},{1.15*sin(25)}) {$S_v$};
                \draw[fill=gray!30] ({cos(45},{sin(45)}) to ({cos(45},{-sin(45)}) arc(-45:45:1);
                \fill (1,0) circle[radius=.035];
                \node at (.83,-.25) {$P_v$};
                \draw (.5,0) to (1,0);
                \draw (1,0) to node[below right=-.2em] {$r$} ({cos(45},{sin(45)});
                \node[above right] at (1,0) {$v$};
            \end{scope}
        \end{tikzpicture}
        \caption{The shaded region is a spherical cap around $v$ with radius $r$.}
        \label{fig:delta-patch}
    \end{figure}
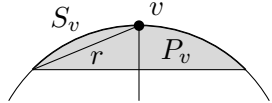

    Next, we prove some properties of spherical caps.

    \begin{lemma}[Properties of spherical caps]
        \label{lem:spherical-patches-properties}
        Let $d \in \N$, $v,w \in S_{d-1}$, $r > 0$, and $P_v$ and $P_w$ the spherical caps with radius $r$ around $v$ and $w$, respectively. Then,
        \begin{enumerate}[nosep]
            \item $P_v = \{x \in B_d : x^Tv > 1-r^2/2\}$.
            \item If $\norm{v-w} \geq 2r$, then $P_v$ and $P_w$ are disjoint.
            \item $\Vol(P_v) \in \Theta(r^{d+1})$, in the limit where $d$ is fixed and $r \downarrow 0$.
        \end{enumerate}
    \end{lemma}

    \begin{proof}
        For the first claim, observe that $x \in S_v$ if and only if $x^Tv = (\norm{x}^2 + \norm{v}^2 - \norm{x-v}^2)/2 > 1-r^2/2$. Consequently, $x \in P_v$ if and only if it is inside the unit ball and satisfies this inequality.

        For the second claim, we prove the contrapositive, i.e., we prove that $z \in P_v \cap P_w \Rightarrow \norm{v-w} < 2r$. To that end, let $z \in P_v \cap P_w$. Then, using that $\norm{z} \leq 1$ and Cauchy-Schwarz's inequality, we obtain
        \[\norm{v-w}^2 = 2\norm{v}^2 + 2\norm{w}^2 - \norm{v+w}^2 \leq 4 - \norm{z}^2\norm{v+w}^2 \leq 4 - (z^T(v+w))^2 < 4 - (2 - r^2)^2 \leq 4r^2.\]

        Finally, computing the volume yields
        \[\Vol(P_v) = \Gamma_{d-1}\int_{1-r^2/2}^1 (1-x^2)^{\frac{d-1}{2}} \;\mathrm{d}x = \Gamma_{d-1} \int_0^{r^2/2} (x(2-x))^{\frac{d-1}{2}} \;\mathrm{d}x \in \Theta(r^{d+1}).\qedhere\]
    \end{proof}

    We now use the above properties to compute how many spherical caps we can pack on the surface of a unit ball.

    \begin{lemma}[Packing of spherical caps]
        \label{lem:spherical-cap-packing}
        Let $d \in \N_{\geq2}$ and let $\delta > 0$. Let $\{v_j\}_{j=1}^n$ be a $\delta$-net, for each $j \in [n]$ let $P_j := P_{v_j}$ be a spherical cap with radius $\delta/2$, and let $P = \{P_j\}_{j=1}^n$. Then, $n = \Theta(\delta^{-(d-1)})$, all the spherical caps are disjoint, and the total volume of the spherical caps is $\Vol(P) := n\Vol(P_1) = \Theta(\delta^2)$.
    \end{lemma}

    \begin{proof}
        The claims follow immediately from \cref{lem:num-patches-lb} and \cref{lem:spherical-patches-properties}.
    \end{proof}

    Finally, we show how embedding a bitstring in these spherical caps naturally leads to a lower bound on the query complexity of the volume estimation and the convex set estimation problems.

    \begin{theorem}
        \label{thm:lbs}
        Let $d \in \N_{\geq 2}$ and $\varepsilon > 0$. Then, any algorithm $\A$ that solves the convex set estimation problem or the volume estimation problem, must make a number of queries that satisfies the complexities listed in \cref{tbl:results}.
    \end{theorem}

    \begin{proof}
        We start with the lower bounds on the convex set estimation problems. To that end, suppose that an algorithm $\A$ solves it with relative Nikodym distance $\varepsilon$.

        For every $\delta > 0$, we choose a particular $\delta$-net $\{v_j\}_{j=1}^n$, and then we let $\{P_j\}_{j=1}^n$ and $P$ be as in \cref{lem:spherical-cap-packing}. Observe that $\Vol(P) \in \Theta(\delta^2)$. We now choose the smallest $\delta > 0$ such that $\Vol(P) \geq 8\varepsilon\Gamma_d$, which for small enough $\varepsilon > 0$ is guaranteed to be well-defined. Then $\delta \in \Theta(\sqrt{\varepsilon})$, which implies that $n \in \Theta(\varepsilon^{-(d-1)/2})$.

        Now, for any bitstring $x \in \{0,1\}^n$, we define a convex body $K_x$ by
        \begin{equation}
            \label{eq:hard-instances}
            K_x = K_0 \cup \bigcup_{\substack{j = 1 \\ x_j = 1}}^n P_j,
        \end{equation}
        Then, given any bitstring $x \in \{0,1\}^n$, $\A$ outputs $\widetilde{K}$, an approximation of $K_x$ with relative precision $\varepsilon$. Next, we can define a bitstring $z \in \{0,1\}^n$ such that
        \[\Vol(\widetilde{K} \Delta K_z) = \min\{\Vol(\widetilde{K} \Delta K_y) : y \in \{0,1\}^n\} \leq \Vol(\widetilde{K} \Delta K_x) \leq \varepsilon\Vol(K_x) \leq \varepsilon\Gamma_d.\]
        Then, we obtain that
        \[|x \oplus z| = \frac{\Vol(K_z \Delta K_x)}{\Vol(P_1)} \leq \frac{n(\Vol(\widetilde{K} \Delta K_z) + \Vol(\widetilde{K} \Delta K_x))}{\Vol(P)} \leq \frac{n(\varepsilon\Gamma_d + \varepsilon\Gamma_d)}{8\varepsilon\Gamma_d} = \frac{n}{4}.\]
        Finally, we observe that any call to the membership oracle of $K_x$ can be simulated with at most one query to the bitstring $x$. Thus, according to \cref{thm:bit-string-recovery}, $\A$ must make at least $\Omega(n) = \Omega(\varepsilon^{-(d-1)/2})$ queries to the membership oracle of $K_x$, in the deterministic, randomized and quantum settings. Finally, since the $\varepsilon$-kernel construction problem is harder than the $\varepsilon$-Nikodym construction problem, the lower bounds from the first two columns in \cref{tbl:results} follow.

        It remains to prove the lower bounds for volume estimation. To that end, suppose that $\A$ solves the volume estimation problem up to precision $\varepsilon$. Again for every $\delta > 0$ we choose a particular $\delta$-net $\{v_j\}_{j=1}^n$, and we let $\{P_j\}_{j=1}^n$ and $P$ as in \cref{lem:spherical-cap-packing}, which means that $\Vol(P) \in \Theta(\delta^2)$. Next, we let $\varepsilon' > 0$ arbitrarily and choose the smallest $\delta > 0$ such that $\Vol(P) \geq 8\varepsilon'\Gamma_d$. For small enough $\varepsilon'$, this is well-defined, and we find that $\delta \in \Theta(\sqrt{\varepsilon'})$, from which we find that $n \in \Theta((\varepsilon')^{-(d-1)/2})$.

        We let $x \in \{0,1\}^n$ and $K_x$ as in \cref{eq:hard-instances}. Then,
        \[\Vol(K_x) = \Vol(K_0) + |x|\Vol(P_1).\]
        $\A$ outputs an estimate $\widetilde{V}$ of $\Vol(K_x)$, such that $|\widetilde{V} - \Vol(K_x)| \leq \varepsilon\Vol(K_x) \leq \varepsilon\Gamma_d$. Hence, we can compute $w = (\widetilde{V} - \Vol(K_0))/\Vol(P_1)$, and we observe that
        \[|w - |x|| = \left|\frac{\widetilde{V} - \Vol(K_0)}{\Vol(P_1)} - \frac{\Vol(K_x) - \Vol(K_0)}{\Vol(P_1)}\right| = \frac{n|\widetilde{V} - \Vol(K_x)|}{\Vol(P)} \leq \frac{n\varepsilon\Gamma_d}{8\varepsilon'\Gamma_d} = \frac{n\varepsilon}{8\varepsilon'} =: k.\]
        Finally, similarly as before, since we can simulate any query to a membership oracle to $K_x$ by one query to the bitstring $x$, we can use the lower bounds from \cref{thm:hamming-weight-estimation} to lower bound the query complexity of $\A$.

        For the lower bounds of \cref{thm:hamming-weight-estimation} to apply, we have to choose $\varepsilon' > 0$ as a function of $\varepsilon$ such that $1 \leq k \leq n/4$. To that end, we let $\ell = \varepsilon/2$ and $u = n\varepsilon/8$, and observe that we must choose $\varepsilon' \in [\ell,u]$. We choose $\varepsilon'$ as $\ell$, $\sqrt{\ell u}$ and $u$, in the deterministic, randomized, and quantum settings, respectively. We remark in the randomized setting that $n/k = 8\varepsilon'/\varepsilon \in \Theta(\sqrt{n})$, and so the lower bound from \cref{thm:hamming-weight-estimation} becomes~$\Omega(n)$. Similarly, in the quantum case, we have $k \in \Theta(1)$, and so the lower bound also becomes~$\Omega(n)$. It remains to plug in the asymptotic scaling of $n$, from which we obtain the following lower bounds for the volume estimation problem:
        \begin{align*}
            \textrm{Deterministic:}\quad & \Omega(n) = \Omega((\varepsilon')^{-\frac{d-1}{2}}) = \Omega(\varepsilon^{-\frac{d-1}{2}}) \\
            \textrm{Randomized:}\quad & \Omega(n) = \Omega(n^{\frac{4}{d+3}} \cdot n^{\frac{d-1}{d+3}}) = \Omega((\varepsilon')^{-\frac{2(d-1)}{d+3}} \cdot (\varepsilon'/\varepsilon)^{\frac{2(d-1)}{d+3}}) = \Omega(\varepsilon^{-\frac{2(d-1)}{d+3}}). \\
            \textrm{Quantum:}\quad & \Omega(n) = \Omega(n^{\frac{2}{d+1}} \cdot n^{\frac{d-1}{d+1}}) = \Omega((\varepsilon')^{-\frac{d-1}{d+1}} \cdot (\varepsilon'/\varepsilon)^{\frac{d-1}{d+1}}) = \Omega(\varepsilon^{-\frac{d-1}{d+1}}).\qedhere
        \end{align*}
    \end{proof}

    \section*{Acknowledgements}

    We would like to thank anonymous reviewers for directing our attention to the existing literature on $\varepsilon$-kernels, and many helpful comments on the presentation of the result.
    SA was supported in part by the European QuantERA project QOPT (ERA-NET Cofund 2022-25), the French PEPR integrated projects EPiQ (ANR-22-PETQ-0007) and HQI (ANR-22-PNCQ-0002), and the French ANR project QUOPS (ANR-22-CE47-0003-01).
    AC was supported by a Simons-CIQC postdoctoral fellowship through NSF QLCI Grant No.\ 2016245.

    \bibliographystyle{alpha}
    \bibliography{arxiv-references}

    \appendix

    \section{Query complexities of bitstring problems}
    \label{app:bit-string-problems}

    In this appendix, we provide the proofs to the query complexities of the bitstring problems that we use in the lower bound constructions in \cref{sec:lower-bounds}. We restate the theorems for convenience.

    \bitstringrecovery*

    \begin{proof}
        The (deterministic) algorithm is trivial -- simply query all the entries of $x$ and output the bitstring exactly. Since the randomized and quantum settings can simulate the deterministic setting, this provides the upper bounds on the query complexities.

        It remains to prove that this algorithm is optimal in all three models up to constant multiplicative factors. Since we can simulate any deterministic and randomized algorithm in the quantum setting, proving hardness in the latter suffices. An information-theoretic argument tells us that quantumly we need at least $\Omega(n)$ queries to output a bitstring that differs from $x$ in at most $n/4$ positions. The core technique stems from \cite{farhi1999bound}, and all remaining details can for instance be found in \cite[Lemma~4.6]{CHJ22}.
    \end{proof}

    \hammingweightestimation*

    \begin{proof}
        We first prove the upper bounds. The deterministic upper bound is trivial. In the randomized setting, we need to show that the query complexity is $O((n/k)^2)$. To that end, suppose that we sample $3(n/k)^2$ bits at random, with replacement, and output $n$ times the fraction of $1$'s observed. Let $X_j$ be the random variable describing the output of the $j$th sample, and we write the output of the algorithm as $w := n/3 \cdot (k/n)^2 \cdot \sum_{j=1}^{3(n/k)^2} X_j$. Then, we trivially find that $\E[w] = |x|$, and its variance can be bounded as
        \[\Var[w] = \Var\left[\frac{n}{3}\left(\frac{k}{n}\right)^2 \cdot \sum_{j=1}^{3(n/k)^2} X_j\right] = \frac{n^2}{9}\left(\frac{k}{n}\right)^4 \cdot 3\left(\frac{n}{k}\right)^2 \Var[X_1] = \frac13k^2\Var[X_1] \leq \frac13k^2,\]
        and so by Chebyshev's inequality, we have
        \[\P[|w - |x|| \geq k] \leq \frac{\Var[w]}{k^2} \leq \frac13.\]
        Finally, quantumly, the algorithm is known as the quantum approximate counting algorithm \cite{BHMT98}, and the upper bound follows directly.

        Then, it remains to prove the lower bounds. Deterministically, the problem is easiest when $k$ is largest, so it suffices to prove the lower bound when $k = n/4$. Once we have queried $n/2-1$ bits of the bitstring -- let's say we observed $w$ ones -- there are still $n/2+1$ bits left to query. Then, we know that the true Hamming weight must be between $w$ and $w + n/2 + 1$. Since this interval is larger than $2k$, we cannot guarantee that whatever number we output is at most $k$ away from the true Hamming weight. Thus, we need to make at least $n/2 \in \Omega(n)$ queries.

        In the randomized setting, we use \cite[Lemma~26]{BDB20}. If $k = \sqrt{n}$, we need at least $\Omega(n)$ queries. Since the problem becomes more difficult when $k$ decreases, this proves the lower bound for all $k \leq \sqrt{n}$. On the other hand, if $k > \sqrt{n}$, then we define $n' \in \Theta((n/k)^2)$. For any bitstring $x' \in \{0,1\}^{n'}$, we define a bitstring $x$ of length $\Theta(n)$ by duplicating all the bits of $x'$ a total of $n/n' \in \Theta(k^2/n)$ times. Then, a Hamming weight estimation algorithm with precision $k$ will find an approximation of the Hamming weight of $x$ with precision $k$, which means that it outputs an approximation of the Hamming weight of $x'$ with precision $\Theta(k/(k^2/n)) = \Theta(n/k) = \Theta(\sqrt{n'})$. Thus, this algorithm must query $x$ at least $\Omega(n') = \Omega((n/k)^2)$ times.

        Finally, in the quantum setting, the result follows easily from \cite{Amb00}. For $k = 1$, the hardness already follows from the majority function, whose hardness was proved in \cite{BBC+01}, and the hardness of the general case was basically proved in \cite{NW99}, even though not phrased as such.
    \end{proof}
\end{document}